\documentclass{llncs}
\usepackage{amsmath,amsfonts,amssymb,srcltx}
\usepackage{graphicx}

\usepackage[margin=1.0in]{geometry}

\usepackage[footnotesize]{caption}

\usepackage{url}

\newcommand{\N}{\mathcal{N}}
\newcommand{\C}{\mathcal{C}}
\newcommand{\T}{\mathcal{T}}
\renewcommand{\L}{\mathcal{L}}
\newcommand{\R}{\mathcal{R}}
\newcommand{\B}{\mathcal{B}}

\title{Combinatorial Scoring of Phylogenetic Networks\thanks{The work is supported 
by the National Science Foundation under the grant No. IIS-1462107.}}

\author{Nikita Alexeev\quad and\quad Max A. Alekseyev}

\institute{The George Washington University, Washington, D.C., U.S.A.}

\begin{document}

\maketitle 

\thispagestyle{plain}


\vspace{-1em}
\begin{abstract}
Construction of phylogenetic trees and networks for extant species from their characters represents one of the key problems in phylogenomics.
While solution to this problem is not always uniquely defined and there exist multiple methods for tree/network construction, 
it becomes important to measure how well the constructed networks capture the given character relationship across the species.

In the current study, we propose a novel method for measuring the \emph{specificity} of a given phylogenetic network in terms of 
the total number of distributions of character states at the leaves that the network may impose. 
While for binary phylogenetic trees, this number has an exact formula and 
depends only on the number of leaves and character states but not on the tree topology, 
the situation is much more complicated for non-binary trees or networks. 
Nevertheless, we develop an algorithm for combinatorial enumeration of such distributions, 
which is applicable for arbitrary trees and networks under some reasonable assumptions.
\end{abstract}



\section{Introduction} 

The evolutionary history of a set of species is often described with a rooted phylogenetic tree with the species at the leaves and their common ancestor at the root.
Each internal vertex and its outgoing edges in such a tree represent a speciation event followed by independent descents with modifications, which outlines the traditional view of evolution. Phylogenetic trees do not however account for \emph{reticulate events} (i.e., partial merging of ancestor lineages), which may also play a noticeable role in evolution through hybridization, horizontal gene transfer, or recombination~\cite{holland2004,huson2011survey}.
Phylogenetic networks represent a natural generalization of phylogenetic trees to include reticulate events. In particular, phylogenetic networks can often more accurately describe the evolution of characters (e.g., phenotypic traits) observed in the extant species. Since there exists a number of methods for construction of such phylogenetic networks~\cite{ulyantsev2015,wu2013},
it becomes important to measure how well the constructed networks capture the given character relationship across the species.

In the current study, we propose a novel method for measuring the \emph{specificity} of a given phylogenetic network in terms of 
the total number of distributions of character states at the leaves that the network may impose. 
While for binary phylogenetic trees, this number has an exact formula and depends only on the number of leaves and character states but not on the tree topology \cite{steel1992,kelk2015}, the situation is much more complicated for non-binary trees or network. 
Nevertheless, we propose an algorithm for combinatorial enumeration of such distributions, 
which is applicable for arbitrary trees and networks under the assumption that reticulate events do not much interfere with each other as explained below.

We view a \emph{phylogenetic network} $\N$ as a rooted directed acyclic graph (DAG). Let $\N^\star$ be an undirected version of $\N$, 
where the edge directions are ignored. We remark that if there are no reticulate events in the evolution, 
then $\N$ represents a tree and thus $\N^\star$ contains no cycles.
On the other hand, reticulate events in the evolution result in appearance of parallel directed paths in $\N$ and cycles in $\N^\star$. 
In our study, we restrict our attention to \emph{cactus networks} $\N$, for which $\N^\star$ represents a cactus graph, i.e.,
every edge in $\N^\star$ belongs to at most one simple cycle. 
In other words, we require that the simple cycles in $\N^\star$ (resulting from reticulate events) are all pairwise edge-disjoint. 
This restriction can be interpreted as a requirement for reticulate events to appear in ``distant'' parts of the network.
Trivially, trees represent a particular case of cactus networks. 
We remark that some problems, which are NP-hard for general graphs, are polynomial for cactus graphs \cite{korneyenko1994}, 
and some phylogenetic algorithms are also efficient for cactus networks \cite{brandes2009}. 
We also remark that cactus networks generalize galled trees (where cycles are vertex-disjoint) and represent a particular case of galled networks (where cycles may share edges)~\cite{Huson2009}.

We assume that the \emph{leaves} (i.e., vertices of outdegree $0$) of a given phylogenetic network $\N$ represent extant species.
A \emph{$k$-state character} is a partition of the species (i.e., leaves of $\N$) into $k$ nonempty sets.
In this paper, we consider only \emph{homoplasy-free} multi-state characters (see \cite{semple2002tree}), 
and enumerate the possible number of such $k$-state characters for any particular $\N$.


\section{Methods}

Let $\N$ be a cactus network. 
For vertices $u,v$ in $\N$, we say that $u$ is an \emph{ancestor} of $v$ and $v$ is a \emph{descendant} of $u$, denoted $u\succcurlyeq v$, if there exists a path from $u$ to $v$ (possibly of length $0$ when $u=v$). Similarly, we say that $v$ is \emph{lower than} $u$, denoted $u\succ v$, if $u\succcurlyeq v$ and $u\ne v$. For a set of vertices $V$ of $\N$, we define a \emph{lowest common ancestor} as a vertex $u$ in $\N$ such that $u\succcurlyeq v$ for all $v\in V$ and there is no vertex $u'$ in $\N$ such that $u\succ u'\succcurlyeq v$ for all $v\in V$. While for trees a lowest common ancestor is unique for any set of vertices, this may be not the case for networks in general. However, in Section~\ref{sec:networks}, we will show that a lowest common ancestor in a cactus network is also unique for any set of vertices.

A $k$-state character on $\mathcal{N}$ can be viewed as a \emph{$k$-coloring}
on the leaves of $\mathcal{N}$,
i.e., a partition of the leaves into $k$ nonempty subsets, each colored with a unique color numbered from $1$ to $k$ (in an arbitrary order). 
A $k$-state character is homoplasy-free if the corresponding $k$-coloring $\C$ of the leaves of $\N$ is \emph{convex}, i.e., the coloring $\C$ can be expanded to some internal vertices of $\N$ such that the subgraphs induced by the vertices of each color are rooted and connected.

Our goal is to compute the number of homoplasy-free multi-state characters on the leaves of $\N$, which is the same as the total number of convex colorings $p(\N)=\sum_{k=1}^{\infty} p_k(\N)$, 
where $p_k(\mathcal{N})$ is the number of convex $k$-colorings on the leaves of $\N$.

\subsection{Trees}

In this section, we describe an algorithm for computing $p(\mathcal{T})$, where $\mathcal{T}$ is a rooted phylogenetic tree. 

We uniquely expand each convex $k$-coloring of the leaves of $\T$ to a \emph{partial $k$-coloring} of its internal vertices as follows. 
For the set $L_i$ of the leaves of color $i$, we color to the same color $i$ 
their lowest common ancestor $r_i$ and all vertices on the paths from $r_i$ to the leaves in $L_i$ (since the $k$-coloring of leaves is convex, this coloring procedure is well-defined). 
We call such partial $k$-coloring of (the vertices of) $\T$ \emph{minimal}. 
Alternatively, a partial $k$-coloring on $\T$ is minimal if and only if for each $i=1,2,\dots,k$, 
the induced subgraph $\T_i$ of color $i$ is rooted and connected (i.e., forms a subtree of $\T$), and removal of any vertex of $\T_i$ that is not a leaf of $\T$ breaks the connectivity of $\T_i$.

By construction, $p_k(\T)$ equals the number of minimal $k$-colorings of $\T$.

The number $p_k(\T)$ in the case of \emph{binary} trees is known \cite{steel1992,kelk2015} and depends only on the number of leaves in a binary tree $\T$, 
but not on its topology. 

\begin{theorem}[{\cite[Proposition 1]{steel1992}}]
\label{th:fib}
Let $\mathcal{T}$ be a rooted binary tree with $n$ leaves. 
Then the number of convex $k$-colorings of $\mathcal{T}$ is $\binom{2n-k-1}{k-1}$. 
Correspondingly, $p(\mathcal{T})$ equals the Fibonacci number $F_{2n-1}$.
\end{theorem}

The case of arbitrary (non-binary) trees is more sophisticated. 

Let $\T$ be a rooted tree. For a vertex $v$ in $\T$, 
we define $\mathcal{T}_v$ as the \emph{full subtree} of $\mathcal{T}$ rooted at $v$ and containing all descendants of $v$. 

Let $\T'$ be any rooted tree larger than $\T$ such that $\T$ is a full subtree of $\T'$. 
We call a $k$-coloring of $\T$ \emph{semiminimal} if this coloring is induced by some minimal coloring on $\T'$ (which may use more than $k$ colors). Clearly, all minimal colorings are semiminimal, but not all semiminimal colorings are minimal.
 We remark that a semiminimal $k$-coloring of $\T$, in fact, does not depend on the topology of $\T'$ outside $\T$ and thus is well-defined for $\T$.

\begin{lemma}
A semiminimal $k$-coloring of $\T$ is well-defined. \label{lem:semiminimal}
 \end{lemma}
\begin{proof}
 Let $\C'$ be a minimal coloring of $\T'$ and $\C$ be its induced coloring on $\T$.
 
 If $\C'$ is such that $\T$ and $\T' \setminus \T$ have no common colors, then $\C$ is minimal, and this property does not depend on the topology of $\T' \setminus \T$. 
 
 If $\C'$ is such that $\T$ and $\T' \setminus \T$ have some common color $i$, 
 then the root $r$ of $\T$ and its parent in $\T'$ are colored into $i$ (hence, the shared color $i$ is unique). Then the coloring on $\T\cup\{(r,l)\}$, where $\T$ inherits its coloring from $\T'$ and $l$ is a new leaf colored into $i$, is minimal. This property does not depend on the topology of $\T' \setminus \T$ either.
\qed
\end{proof}

For a semiminimal $k$-coloring $\mathcal{C}$ on $\T_v$, there exist three possibilities:
\begin{itemize}
\item Vertex $v$ is colored and shares its color with at least two of its children. 
In this case, coloring $\mathcal{C}$ is minimal. 
\item Vertex $v$ is colored and shares its color with exactly one of its children, 
and coloring $\mathcal{C}$ is not minimal. 
\item Vertex $v$ is not colored. In this case, coloring $\C$ represents a minimal coloring of $\T_v$.
\end{itemize}
Correspondingly, for each vertex $v$, we define
\begin{itemize}
 \item $f_k(v)$ is the number of minimal $k$-colorings of $\mathcal{T}_v$ such that 
 at least two children of $v$ have the same color (the vertex $v$ must also have this color);
 \item $g_k(v)$ is the number of semiminimal $k$-colorings of $\mathcal{T}_v$ such that 
 the vertex $v$ shares its color with exactly one of its children (i.e., semiminimal but not minimal $k$-colorings);
 \item $h_k(v)$ is the number of minimal $k$-colorings of $\mathcal{T}_v$ such that the vertex $v$ is not colored.
\end{itemize}
We remark that the number of minimal $k$-colorings of $\T$ equals $f_k(r)+h_k(r)$, where $r$ is the root of the tree $\T$.

We define the following generating functions:
\begin{equation}\label{eq:FGH}
F_v(x) = \sum_{k=1}^\infty f_k(v)\cdot x^k;\qquad
G_v(x) = \sum_{k=1}^\infty g_k(v)\cdot x^k;\qquad
H_v(x) = \sum_{k=1}^\infty h_k(v)\cdot x^k.
\end{equation}

For a leaf $v$ of $\T$, we assume $f_k(v) = \delta_{k,1}$ (Kronecker's delta) and $g_k(v)=h_k(v)=0$ for any $k\geq 1$. 
Correspondingly, we have $F_v(x) = x$ and $G_v(x) = H_v(x) = 0$.

If a vertex $v$ has $d$ children $u_1,u_2, \dots,u_d$, then one can compute $F_v(x)$, $G_v(x)$, and $H_v(x)$ 
using the generating functions at the children of $v$ as follows.

\begin{theorem}
 For any internal vertex $v$ of $\mathcal{T}$, we have
 \begin{small}
 \begin{align}
  H_v(x) &= \prod_{i=1}^d (F_{u_i}(x) + H_{u_i}(x)); \label{eq:H}\\
  G_v(x) &= \sum_{i=1}^d (F_{u_i}(x) + G_{u_i}(x)) \prod_{j=1\atop j \neq i}^d (F_{u_j}(x) + H_{u_j}(x))
  = H_v(x)\cdot \sum_{i=1}^d \frac{F_{u_i}(x) + G_{u_i}(x)}{F_{u_i}(x) + H_{u_i}(x)};\label{eq:G} \\
  F_v(x) &= x\prod_{i=1}^d \left(F_{u_i}(x) + H_{u_i}(x)+\frac{F_{u_i}(x) + G_{u_i}(x)}{x}\right)-x\cdot H_v(x)-G_v(x); \label{eq:F}
 \end{align}
 \end{small}
where $u_1,\,u_2,\,\dots,\,u_d$ are the children of $v$.
\label{thm:trees}
\end{theorem}
\begin{proof}
Suppose that vertex $v$ is not colored in a minimal $k$-coloring of $\T$. 
Then each its child is either not colored or has a color different from those of the other children. 
Furthermore, if a child of $v$ is colored, its color must appear at least twice among its own children. Thus, 
the number of semiminimal colorings of $\T_{u_i}$ in this case is $f_{k_i}(u_i)+h_{k_i}(u_i)$, where $k_i$ is the number of colors in $\T_{u_i}$ ($i=1,2,\dots,d$).
Also, the subtrees $T_{u_i}$ cannot share any colors with each other.
Hence, the number of minimal $k$-colorings of $\T_v$ with non-colored $v$ equals 
$$\sum_{k_1+\dots+k_d=k} \,\, \prod_{i=1}^d (f_{k_i}(u_i)+h_{k_i}(u_i)),$$ 
implying formula \eqref{eq:H}.

Now suppose that vertex $v$ is colored in a minimal $k$-coloring of $\T$. Then $v$ must share its color with at least one of its children. Consider two cases.

\textbf{Case 1.} Vertex $v$ shares its color with exactly one child, say $u_i$. Then there are $f_{k_i}(u_i)+g_{k_i}(u_i)$ semiminimal $k_i$-colorings 
for $\T_{u_i}$. For any other child $u_j$ ($j\ne i$), similarly to the above, 
we have that the number of semiminimal $k_j$-colorings equals $f_{k_j}(u_j)+h_{k_j}(u_j)$. 
Hence, the number of semiminimal $k$-colorings of $\T_v$ in this case equals 
$$\sum_{k_1+\dots+k_d=k} (f_{k_i}(u_i)+g_{k_i}(u_i)) \prod_{j=1\atop j\ne i}^d (f_{k_j}(u_j)+h_{k_j}(u_j)),$$
implying formula \eqref{eq:G}.

\textbf{Case 2.} Vertex $v$ shares its color with children $u_i$, $i\in I$, $|I|\geq 2$, but not with $u_j$ for $j\notin I$. 
Since the color of $v$ is the only color shared by $T_{u_i}$, we have $k_1+\dots+k_d=k+|I|-1$. Similarly to Case 1,
we get that the number of minimal $k$-colorings of $\mathcal{T}_v$ is a coefficient at $x^{k+|I|-1}$ in
$$\prod_{i\in I} \left(F_{u_i}(x) + G_{u_i}(x)\right)\prod_{j \notin I} \left(F_{u_j}(x) + H_{u_j}(x)\right),$$
which is the same as the coefficient of $x^k$ in
$$x\prod_{i\in I} \frac{\left(F_{u_i}(x) + G_{u_i}(x)\right)}{x}\prod_{j \notin I} \left(F_{u_j}(x) + H_{u_j}(x)\right).$$
Summation of this expression over all subsets $I\subset\{1,2,\dots,d\}$ gives us the first term of \eqref{eq:F}, 
from where we subtract the sum over $I$ with $|I|=0$ (the term $x\cdot H_v(x)$) and
with $|I|=1$ (the term $G_v(x)$) to prove \eqref{eq:F}.
\qed
\end{proof}


\subsection{Cactus Networks}\label{sec:networks}

In this section, we show how to compute $p(\N)$, where $\N$ is a cactus network. 
The following lemma states an important property of cactus networks. 

\begin{lemma}\label{lem:lca}
Let $\N$ be a cactus network. Then for any set of vertices of $\N$, their lowest common ancestor is unique.
\end{lemma}

\begin{proof}
For any set of vertices, there exists at least one common ancestor, which is the root of $\N$. 

It is enough to prove the statement for 2-element sets of vertices. Indeed, if for any pair of vertices their lowest common ancestor is unique, then
in any set of vertices we can replace any pair of vertices with their lowest common ancestor without affecting the lowest common ancestors of the whole set. After a number of such replacements, the set reduces to a single vertex, which represents the unique lowest common ancestor of the original set.

Suppose that for vertices $u_1$ and $u_2$ in $\N$, there exist two lowest common ancestors $r_1$ and $r_2$.
Let $r'$ be a lowest common ancestor of $r_1$ and $r_2$ (clearly, $r'\ne r_1$ and $r'\ne r_2)$, and $P_1$ and $P_2$ be paths from $r'$ to $r_1$ and $r_2$, respectively. Then $P_1$ and $P_2$ are edge-disjoint. Let $Q_{i,j}$ be paths from $r_i$ to $u_j$ ($i,j=1,2$). It easy to see that for each $i=1,2$, the paths $Q_{i,1}$ and $Q_{i,2}$ are edge-disjoint.
Then the paths $P_1$, $Q_{1,1}$, $P_2$, $Q_{2,1}$ form a simple cycle in $\N^*$; similarly, the paths $P_1$, $Q_{1,2}$, $P_2$, $Q_{2,2}$ form a simple cycle in $\N^*$. These simple cycles share the path $P_1$ (and $P_2$), 
a contradiction to $\N$ being a cactus network.
\qed
\end{proof}

In contrast to trees, cactus networks may contain branching paths, which end at vertices of indegree $2$ called \emph{sinks} (also known as \emph{reticulate vertices}~\cite{Huson2009,huson2011survey}). 
Clearly, there are no vertices of indegree greater than $2$ in a cactus network (if there are three incoming edges to some vertex then each of them belongs to two simple cycles in $\N^\star$). We will need the following lemma.

\begin{lemma}\label{lem:nosinks}
Let $p_l$ and $p_r$ be the parents of some sink in a cactus network $\N$, and $s$ be their lowest common ancestor. 
Let $P_l$ and $P_r$ be paths from $s$ to $p_l$ and $p_r$, respectively. Then $P_l$ and $P_r$ (i) are edge-disjoint; (ii) do not contain sinks, except possibly vertex $s$; 
and (iii) are unique.
\end{lemma}

\begin{proof}
 Since $s$ is the lowest common ancestor of $p_l$ and $p_r$, the paths $P_l$ and $P_r$ are edge-disjoint. 

Suppose that there is an edge $(u,t')$ on a path from $s$ to $p_l$ such that $t'$ is a sink. 
Then this edge belongs to two different simple cycles in $\N^*$, a contradiction to $\N$ being a cactus network.

It is easy to see that if there exists a path $P'_l$ from $s$ to $p_l$ different from $P_l$, then the paths $P_l$ and $P'_l$ would share a sink (different from $s$), which does not exist on $P_l$. Hence, the path $P_l$ is unique and so is $P_r$.
\qed
\end{proof}

Let $t$ be a sink in $\N$ and $p_l$ and $p_r$ be its parents. Let $s$ be the lowest common ancestor of $p_l$ and $p_r$ (which exists by Lemma~\ref{lem:lca}). 
Lemma~\ref{lem:nosinks} implies that the paths $P_1$ and $P_2$ from $s$ to $t$ that visit vertices $p_l$ and $p_r$, respectively, are unique and edge-disjoint. 
We call such a vertex $s$ \emph{source} and refer to the unordered pair of paths $\{P_1,P_2\}$ as a \emph{simple branching path} 
(denoted $s \rightrightarrows t$) and to each of these paths as a \emph{branch} of $s \rightrightarrows t$. 
Notice that one source may correspond to two or more sinks in $\N$.

For any two vertices $u$ and $v$ connected with a unique path in $\N$, we denote this path by $u \rightarrow v$ (which is a null path if $v=u$).
A \emph{branching path} between vertices $p$ and $q$ in $\N$ is an alternating sequence of unique and simple branching paths
$$p \rightarrow s_1 \rightrightarrows t_1 \rightarrow \dots \rightarrow s_m \rightrightarrows t_m \rightarrow q,$$ 
where some of the unique paths may be null.

\begin{lemma}\label{lem:bpath}
For any vertices $u \succcurlyeq v$ in a cactus network $\N$, the union of all paths between them forms a branching path.
\end{lemma}

\begin{proof}
 Suppose $\N$ is a cactus network and $u \succcurlyeq v$ in $\N$. Let $\N'$ be a subnetwork of $\N$ formed by all paths from $u$ to $v$. Clearly, $\N'$ is a rooted (at $u$) cactus. We will prove that $\N'$ is a branching path by induction on the number of sinks in $\N'$. If there are no sinks in $\N'$, then a path from $u$ to $v$ is unique, then the statement holds. Otherwise, there exists a sink $t$ in $\N'$ such that $u \succ t \succcurlyeq v$ and a path from $t$ to $v$ is unique, while there exist multiple paths from $u$ to $t$. Let $s$ be the source in $\N'$ corresponding to $t$. Then the branching path from $s$ to $v$ has the form $s\rightrightarrows t \rightarrow v$. Since every path from $u$ to $v$ visits $t$, it also must visit $s$ (by Lemma~\ref{lem:nosinks} a path cannot enter into a simple branching path $s\rightrightarrows t$ other than through vertex $s$). 
Let $\N''$ be the subnetwork of $\N'$ consisting of all paths from $u$ to $s$. Since the number of sinks in $\N''$ is one less than in $\N'$, by induction it is a branching path. Then $\N'$ is a branching path obtained from $\N''$ by concatenating it with the path $s\rightrightarrows t \rightarrow v$.
\qed
\end{proof}

For vertices $u \succcurlyeq v$ in $\N$, the union of all paths from $u$ to $v$ is called the \emph{maximal branching path}.


We generalize the notion of the minimal $k$-coloring to the case of cactus networks by expanding any convex $k$-coloring of the leaves of $\N$ to a partial coloring of the internal vertices of $\N$ as follows.

\paragraph{Network Coloring Procedure.} For a given convex $k$-coloring of the leaves of $\N$, let $L_i$ be the set of the leaves of color $i$. 
We consider maximal branching paths from the lowest common ancestor $r^{(i)}$ of $L_i$ to all $l\in L_i$, which by Lemma~\ref{lem:bpath} have the form:
$$r^{(i)} \rightarrow s_1^{(i)} \rightrightarrows t^{(i)}_1 \rightarrow \dots \rightarrow s^{(i)}_{m_i} \rightrightarrows t^{(i)}_{m_i} \rightarrow l.$$
At the first step, we color all vertices in the unique subpaths of such maximal branching paths into color $i$ (Fig.~\ref{fig:min}a,b). 
Lemma~\ref{lem:minC} below shows that this coloring procedure (performed for all $i=1,2,\dots,k$) is well-defined.
At the second step, we color some branches of simple branching subpaths of the maximal branching paths from $r^{(i)}$ to the leaves in $L_i$.
Namely, for each branch between $s^{(i)}_j$ and $t^{(i)}_j$ we check if its vertices are colored (at the first step) in any color other than $i$; 
if no other color besides $i$ is present in the branch, we color all its vertices into $i$ (Fig.~\ref{fig:min}b,c). 
Lemma~\ref{lem:minC} below shows that at least one branch of each simple branching subpath is colored this way, implying that the induced subgraphs of each color in the resulting partial coloring are connected.
We refer to the resulting partial coloring as a \emph{minimal $k$-coloring} of $\N$.
  \begin{figure}[!t]
   \begin{center}
   \includegraphics[width = \textwidth]{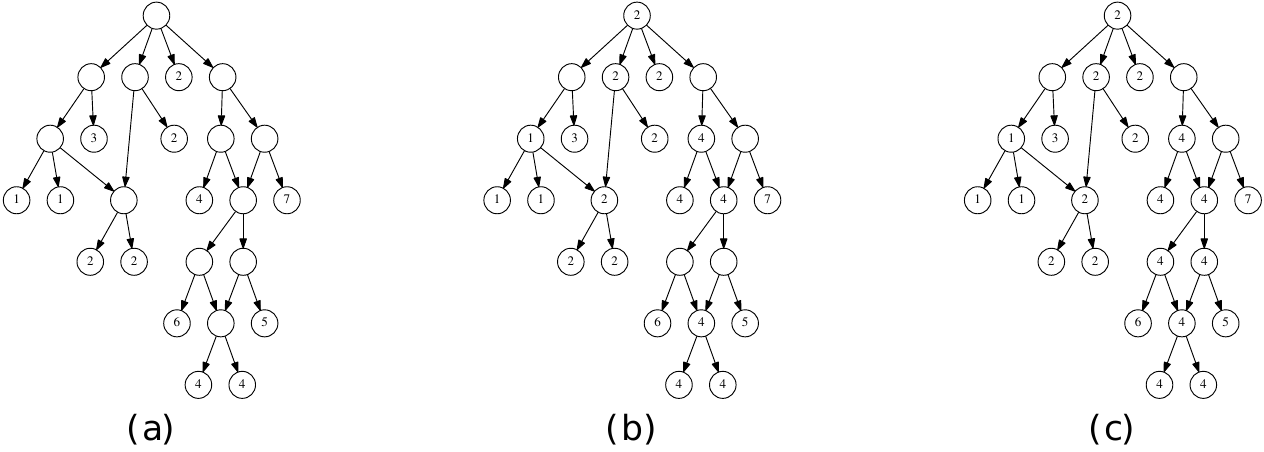}
   \caption{\textbf{(a)} A convex coloring of the leaves of a cactus network $\N$, where the colors are denoted by labels.
   \textbf{(b)} The partial coloring of $\N$ constructed at the first step of the coloring procedure. 
   \textbf{(c)} The minimal coloring of $\N$.}
   \label{fig:min}
   \end{center}
 \end{figure}
\begin{lemma}\label{lem:minC} 
For any convex $k$-coloring on the leaves of a cactus network $\N$, 
the corresponding minimal $k$-coloring of $\N$ is well-defined. Moreover, the induced subgraph of $\N$ of each color is connected.
\end{lemma}

\begin{proof}
 By the definition of convexity, a given convex $k$-coloring of the leaves of $\N$ can be expanded to a partial $k$-coloring $\mathcal{C}$ of $\N$ such that the induced subgraphs of each color are connected. The partial coloring of $\N$ that we obtain at the first step is a subcoloring of $\mathcal{C}$.
Indeed, since the induced subgraphs of each color in $\mathcal{C}$ are connected, the unique subpaths of the maximal branching paths are colored (at the first step) into the same color as in $\mathcal{C}$. Hence, no conflicting colors can be imposed at the first step.

On the second step, we color a branch in some color $i$ only if corresponding source and sink are colored in color $i$, so there are no conflicts on the second step. By Lemma~\ref{lem:nosinks}, each non-source vertex belongs to at most one simple branching path, and thus the second step and the whole coloring procedure are well-defined.

For any simple branching subpath $s^{(i)}_j  \rightrightarrows t^{(i)}_j$, at least one branch, say $b$, is colored into $i$ in $\mathcal{C}$. 
In the subcoloring of $\mathcal{C}$ obtained at the first step, the branch $b$ cannot contain any colors besides $i$.
Hence, we will color all vertices of $b$ into $i$ at the second step. That is, at least one branch of every simple branching subpath will be colored, implying that the induced subgraph of each color is connected.
\qed
\end{proof}

Similarly to the case of trees, we compute $p(\N)$ as the number of minimal colorings of a cactus network $\N$.

Let $\N'$ be any rooted network larger than $\N$ such that $\N$ is a rooted subnetwork of $\N'$ and all edges from 
$\N' \setminus \N$ to $\N$ end at the root of $\N$. 
We call a $k$-coloring of $\N$ \emph{semiminimal} if this coloring is induced by some minimal coloring on $\N'$.
Similarly, to the case of trees (Lemma~\ref{lem:semiminimal}), a semiminimal $k$-coloring of $\N$ does not depend on the topology of $\N' \setminus \N$ and thus is well-defined.

For each vertex $v$ in $\N$, we define a subnetwork $\N_v$ of $\N$ rooted at $v$ and containing all descendants of $v$.
An internal vertex in $\N$ is \emph{regular} if the subnetworks rooted at its children are pairwise vertex-disjoint.
It is easy to see that sources in $\N$ are not regular.

For each vertex $v$ in $\N$, we define the following quantities:
\begin{itemize}
 \item $f_k(v)$ is the number of minimal $k$-colorings of $\N_v$ such that $v$ is colored (\emph{$f$-type coloring});
 \item $g_k(v)$ is the number of semiminimal but not minimal $k$-colorings of $\N_v$ (\emph{$g$-type coloring});
 \item $h_k(v)$ is the number of minimal $k$-colorings of $\N_v$ such that the vertex $v$ is not colored (\emph{$h$-type coloring}).
\end{itemize}
Our goal is to compute $p_k(\N)=f_k(r)+h_k(r)$ for each positive integer $k$, where $r$ is the root of $\N$.
As before, for a leaf $v$ of $\N$, we have $f_k(v)=\delta_{k,1}$ and $g_k(v)=h_k(v)=0$ for any $k\geq 1$.
We define the generating function $F_v(x)$, $G_v(x)$, and $H_v(x)$ as in \eqref{eq:FGH}. Whenever we compute these functions 
in a subnetwork $\mathcal{M}$ of $\N$, we refer to them as $F^\mathcal{M}_v(x)$, $G^\mathcal{M}_v(x)$, and $H^\mathcal{M}_v(x)$.

It is easy to see that Theorem~\ref{thm:trees} holds for all regular vertices $v$ of $\N$ and 
therefore gives us a way to compute $F_v(x)$, $G_v(x)$, 
and $H_v(x)$, provided that these functions are already computed at the children of $v$.
So it remains to describe how to compute these functions at the sources in $\N$.

Let $s$ be a source in $\N$ and $t$ be any sink corresponding to $s$. 
We define $p_l$ and $p_r$ be the parents of $t$. 
To obtain formulas for $F_s(x), G_s(x)$, and $H_s(x)$, 
we consider the auxiliary subnetworks $\N_s$, $\N_t$, $\N_{s\setminus t}=\N_s\setminus\N_t$, $\mathcal{L}$, $\mathcal{R}$, and $\mathcal{B}$ (Fig.~\ref{fig:auxgraph}), where
\begin{itemize}
 \item $\mathcal{L}$ is the subnetwork obtained from $\N_s$ by removing the edge $(p_l,t)$;
 \item $\mathcal{R}$ is the subnetwork obtained from $\N_s$ by removing the edge $(p_r,t)$;
 \item $\mathcal{B}$ is the subnetwork obtained from $\N_s$ by removing all the edges in the simple branching path $s \rightrightarrows t$.
\end{itemize}
 \begin{figure}[!t]
   \begin{center}
   \includegraphics[width = \textwidth]{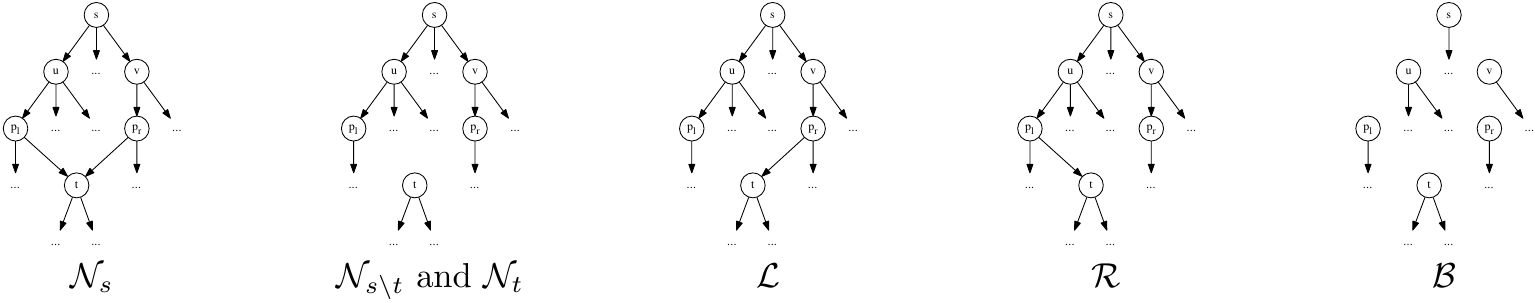}
   \caption{Subnetworks $\N_s$, $\N_{s\setminus t}$, $\N_t$, $\mathcal{L}$, $\mathcal{R}$, and $\mathcal{B}$.}
   \label{fig:auxgraph}
   \end{center}
 \end{figure}
It is easy to see that the vertex sets of the subnetworks $\N_s$, $\mathcal{L}$, $\mathcal{R}$, $\N_t\cup \N_{s\setminus t}$, and $\mathcal{B}$ coincide, and therefore a partial coloring of one subnetwork translates to the others. 

\begin{lemma}\label{lem:colLR}
Let $t$ be a sink in a cactus network $\N$ and $s$ be the corresponding source. 
If $\C$ is a partial coloring of $\N_s$ of $f$-type, $g$-type, or $h$-type, then $\C$ contains a partial subcoloring of $\L$ or $\R$ of the same type.
\end{lemma}

\begin{proof}
 We say that a partial coloring of $\N$ \emph{uses} an edge $(u,v)$ if vertices $u$ and $v$ are colored into the same color. Let $p_l$ and $p_r$ be the parents of $t$. 
Note that if $\C$ does not use the edge $(p_l,t)$ then it is a partial coloring of the same type on $\L$. 
Similarly, if $\C$ does not use the edge $(p_r,t)$ then it is a partial coloring of the same type on $\R$. So, it remains to consider the case when $\C$ uses both edges $(p_l,t)$ and $(p_r,t)$. 

Suppose that $\C$ uses both edges $(p_l,t)$ and $(p_r,t)$. 
Notice that such $\C$ cannot be of $h$-type (since $p_r$ and $p_l$ share the same color, their lowest common ancestor $s$ has to be colored as well). 
So, $\C$ has $f$-type or $g$-type. Let $i$ be the color of $t$. From the second step of the coloring procedure, it follows that each vertex $v$ such that $s\succcurlyeq v\succ t$ is also colored into $i$. Hence, removal of one of the edges $(p_l,t)$ and $(p_r,t)$ does not break the connectivity of the induced subgraph of color $i$. 
Thus, if $\C$ has $f$-type, then it contains a partial subcoloring of both $\L$ and $\R$ of $f$-type.
Now suppose that $\C$ has $g$-type and $\C'$ is a subcoloring of $\C$ constructed at the first step of the network coloring procedure. Then at least one branch in $s \rightrightarrows t$ does not contain vertices of color $i$ in $\C'$ (otherwise $\C$ would have $f$-type). If this branch contains $p_l$ then $\C$ contains a partial subcoloring of $\L$ of $g$-type; otherwise $\C$ contains a partial subcoloring of $\R$ of $g$-type.
\qed
\end{proof}

\begin{theorem}\label{thm:networks}
 Let $s$ be a source in $\N$ and $t$ be any sink corresponding to $s$. Then
 \begin{small}
 \begin{align}
  H_s(x) &= H^{\mathcal{L}}_s(x)+ H^{\mathcal{R}}_s(x) - H^{\N_{s\setminus t}}_s(x)\cdot (F_{t}(x) + H_{t}(x)); \label{eq:Hs}\\
  G_s(x) &= G^{\mathcal{L}}_s(x)+ G^{\mathcal{R}}_s(x) \nonumber\\
  &- G^{\N_{s\setminus t}}_s(x)\cdot (F_{t}(x) + H_{t}(x)) - (F_{t}(x) + G_{t}(x))\cdot\prod_{v:\ s\succcurlyeq v \succ t} H^{\mathcal{B}}_v(x); \label{eq:Gs}\\
  F_s(x) &= F^{\mathcal{L}}_s(x)+ F^{\mathcal{R}}_s(x) - F^{\N_{s\setminus t}}_s(x)\cdot (F_{t}(x) + H_{t}(x)) \nonumber \\ 
  &- (F_{t}(x) + G_{t}(x))\cdot \left(\prod_{v:\ s \succcurlyeq v\succ t} \left(\frac{F^{\mathcal{B}}_v(x)+G^{\mathcal{B}}_v(x)}{x}+H^{\mathcal{B}}_v(x)\right)-\prod_{v:\ s\succcurlyeq v\succ t} H^{\mathcal{B}}_v(x)\right) \label{eq:Fs}  
\end{align} 
\end{small}
under the following convention: if a non-leaf vertex $v$ in $\N$ turns into a leaf 
in a network $\N'\in\{\N_{s\setminus t}, \mathcal{L}, \mathcal{R}, \mathcal{B}\}$, then we re-define
$F_v^{\N'}(x)=G_v^{\N'}(x)=0$ and $H_v^{\N'}(x)=1$.
\end{theorem}

\begin{proof}
We say that a partial coloring of $\N$ \emph{uses} an edge $(u,v)$ if vertices $u$ and $v$ are colored into the same color. 
Let $p_l$ and $p_r$ be the parents of $t$.

 
 Let us enumerate $h$-type colorings of $\N_s$ first. We remark that such coloring cannot use both edges $(p_l,t)$ and $(p_r,t)$ (if it uses both these edges, the source $s$ would be colored by the definition of minimal coloring). That is, any $h$-type coloring of $\N_s$ represents an $h$-type coloring of $\L$ or $\R$, or both these networks. The number of $h$-type $k$-colorings of $\L$ and $\R$ is the coefficient of $x^{k}$ in $H^{\L}_s(x)$ and $H^{\R}_s(x)$, respectively. 
 By the inclusion-exclusion principle, the number of $h$-type $k$-colorings of $\N_s$ equals the sum of those of $\L$ or $\R$ minus the number of $h$-type $k$-colorings of both $\L$ and $\R$. A coloring of the last kind does not use either of the edges $(p_l,t)$ and $(p_r,t)$, and thus is formed by an $h$-type coloring of $\N_{s\setminus t}$ and a minimal coloring of $\N_t$ (the colors of the two colorings are disjoint). The number of such coloring pairs equals the coefficient of $x^{k}$ in $ H^{\N_{s\setminus t}}_s(x)(F_{t}(x) + H_{t}(x))$, which completes the proof of \eqref{eq:Hs}.
 
 We use similar reasoning to prove \eqref{eq:Gs} and \eqref{eq:Fs}. The first two terms in these formulas are similar to those in \eqref{eq:Hs} that correspond to same-type colorings on $\L$ or $\R$ (at least one of which always exists by Lemma~\ref{lem:colLR}). The case of colorings of $g$-type or $f$-type on both $\L$ and $\R$ is more complicated and is split into two subcases depending on whether none or both of the edges $(p_l,t)$ and $(p_r,t)$ are used 
(if exactly one of the edges is used, it cannot be removed without making the induced subgraph of this color disconnected). The subcase of using none of the edges is similar to $h$-type colorings 
and gives us the third term in formulas \eqref{eq:Gs} and \eqref{eq:Fs}. So it remains to enumerate colorings on both $\L$ and $\R$ that use both edges $(p_l,t)$ and $(p_r,t)$.
 
Let $\C$ be a $g$-type $k$-coloring of $\N_s$ that is a coloring on both $\L$ and $\R$ and uses both edges $(p_l,t)$ and $(p_r,t)$. 
Let $\C'$ be a subcoloring of $\C$ constructed at the first step of the network coloring procedure. 
Vertices $s$ and $t$ have the same color in both $\C$ and $\C'$, but any vertex $v$ with $s \succ v \succ t$ is colored in $\C$ but not in $\C'$. So $\C$ corresponds to a semiminimal coloring on $\N_t$ with a colored root $t$ (i.e., of $f$-type or $g$-type) and 
a coloring on $\mathcal{B}\setminus\N_t$ such that vertices $v$ with $s \succcurlyeq v \succ t$ (in $\N_s$) are not colored. 
Since $\B$ is the union of vertex-disjoint subnetworks $\N_v$ with $s \succcurlyeq v \succcurlyeq t$, the number of such coloring pairs equals the coefficient of $x^{k}$ in
 $$(F_{t}(x) + G_{t}(x))\cdot \prod_{v:\ s\succcurlyeq v\succ t} H^{\mathcal{B}}_v(x).$$
 
Now, let $\C$ be an $f$-type $k$-coloring of $\N_s$ that is a coloring on both $\L$ and $\R$ and uses both edges $(p_l,t)$ and $(p_r,t)$.  
 Let $\C'$ be a subcoloring of $\C$ constructed at the first step of the network coloring procedure. 
 Vertices $s$ and  $t$ have the same color in both $\C$ and $\C'$, and any vertex $v$ with $s \succ v \succ t$ is either colored into the same color or not colored in $\C'$. So $\C$ corresponds to a semiminimal coloring on $\N_t$ with the colored root $t$ (i.e., of $f$-type or $g$-type) and a coloring on $\B\setminus\N_t$ such that at least one vertex $v$ with $s \succcurlyeq v \succ t$ (in $\N_s$) is colored into the same color. The number of such colorings is the coefficient of $x^{k}$ in
 $$(F_{t}(x) + G_{t}(x))\cdot \left(\prod_{v:\ s\succcurlyeq v\succ t} \left(\frac{F^{\mathcal{B}}_v(x)+G^{\mathcal{B}}_v(x)}{x}+H^{\mathcal{B}}_v(x)\right)-\prod_{v:\ s\succcurlyeq v\succ t} H^{\mathcal{B}}_v(x)\right).$$
 The first product in the parentheses represents the generating function for the number of colorings of $\B\setminus\N_t$,
 where each vertex $v$ with $s \succcurlyeq v \succ t$ (in $\N_s$) is either colored 
 in a reserved color (accounted by the term $\frac{F^{\mathcal{B}}_v(x)+G^{\mathcal{B}}_v(x)}{x}$) 
 or not colored (accounted by the term $H^{\mathcal{B}}_v(x)$). 
 Subtraction of the second product eliminates the case where no vertex $v$ with $s \succcurlyeq v \succ t$ (in $\N_s$) is colored.
\qed
\end{proof}

\section{Algorithm for Computing $p(\N)$}

Theorem~\ref{thm:trees} (for regular vertices) and Theorem~\ref{thm:networks} (for sources) allow us to compute the generating functions $F,G,H$ at the root $r$ of a cactus network $\N$ recursively. Namely, to compute $F_v(x)$, $G_v(x)$, and $H_v(x)$ for a vertex $v$ (starting at $v=r$), we proceed as follows:
\begin{itemize}
\item
if $v$ is a leaf, then $F_v(x)=x$, $G_v(x)=H_v(x)=0$ (except for the special case of a newly formed leaf described in Theorem~\ref{thm:networks}, when $F_v(x)=G_v(x)=0$ and $H_v(x)=1$);
\item
if $v$ is regular, we recursively proceed with computing $F_u(x)$, $G_u(x)$, and $H_u(x)$ for every child $u$ of $v$, and then combine the results with formulae \eqref{eq:H}, \eqref{eq:G}, \eqref{eq:F};
\item if $v$ is a source, we select any sink $t$ corresponding to $s$, and apply Theorem~\ref{thm:networks} to compute $F_v(x)$, $G_v(x)$, and $H_v(x)$ from the generating functions computed in smaller subnetworks. We remark that while $s$ may be a source for more than one sink and thus may still remain a source in the subnetworks, the number of sinks in each of the subnetworks decreases as compared to $\N$, implying that our recursion sooner or later will turn $s$ into a regular vertex and then recursively proceed down to its children.
\end{itemize}

From the generating functions at the root $r$ of $\N$, we can easily obtain the number $p_k(\N)$ of convex $k$-colorings of $\N$
as the coefficient of $x^k$ in $F_r(x)+H_r(x)$. This further implies that $p(\N)$ can be computed as
$$p(\N) = \sum_{k=0}^\infty p_k(\N) = F_r(1)+H_r(1).$$

In the Appendix, we provide a SAGE~\cite{sage} implementation of the described algorithm for computing functions $F_v(x)$, $G_v(x)$, and $H_v(x)$ for a given cactus network $N$ and its vertex $v$.

\section{Applications}

\paragraph{Network Specificity.}
We propose to measure the specificity of a cactus network $\N$ with $n$ leaves as a decreasing function of $p(\N)$.
Notice that the value of $p(\N)$ can be as small as $2^n-n$ (for a tree with $n$ leaves all being children of the root) 
and as large as the Bell number $B_n$ (enumerating set partitions of the leaves). We therefore find it convenient to define the \emph{specificity score} of $\N$ as
$$\tau(\N) = \frac{n}{\log_2(p(\N)+n)}.$$
In particular, we always have $0<\tau(\N)\leq 1$, where the upper bound is achievable. The asymptotic of $B_n$ further implies that $\tau(\N)$ can be asymptotically as low as $\frac{\log 2}{\log n}$, which vanishes as $n$ grows.

From Theorem~\ref{th:fib}, it can be easily seen that for a binary tree $\T$ with $n$ leaves, 
we have $\tau(\T)\approx \frac{n}{\log_2 \phi^{2n-1}} \approx 0.72$ when $n$ is large, where $\phi=\frac{1+\sqrt{5}}2$ is the golden ratio.

\paragraph{Network Comparison.}
Existing methods for construction of phylogenetic networks (e.g., hybridization networks from a given set of gene trees~\cite{wu2013,ulyantsev2015}) often rely on the parsimony assumption and attempt
to minimize the number of reticulate events. 
Such methods may generate multiple equally parsimonious networks, which will then need to be evaluated and compared from a different perspective.
It is equally important to compare phylogenetic networks constructed by different methods.
If the number of reticulate events in a constructed network is small, it is quite likely that this network represents a cactus network.
Furthermore, there exist methods that explicitly construct phylogenetic cactus networks~\cite{brandes2009}.
This makes our method well applicable for evaluation and comparison of such networks in terms of their specificity as defined above.



\paragraph{Orientation of Undirected Networks.}

Some researchers consider undirected phylogenetic networks (called ``abstract'' in the survey~\cite{huson2011survey}) that describe evolutionary relationship of multiple species but do not correlate their evolution with time.
For a given undirected cactus network $\N'$, our method allows one to find a root and an orientation of $\N'$, i.e., a directed rooted network $\N$ with $\N^\star=\N'$, that maximizes the specificity score $\tau(\N)$. Indeed, different orientations of the same undirected network may result in different scores even if they are rooted at the same vertex. For example, in Fig.~\ref{fig:orient} the network $\N_1$ 
 has $p(\N_1)=35$ convex colorings and the score $\tau(\N_1)\approx 0.94$, while the network $\N_2$ has $p(\N_2) = 37$ convex colorings and the score $\tau(\N_2) \approx 0.927$.
 \begin{figure}[!t]
   \begin{center}
   \includegraphics[width = \textwidth]{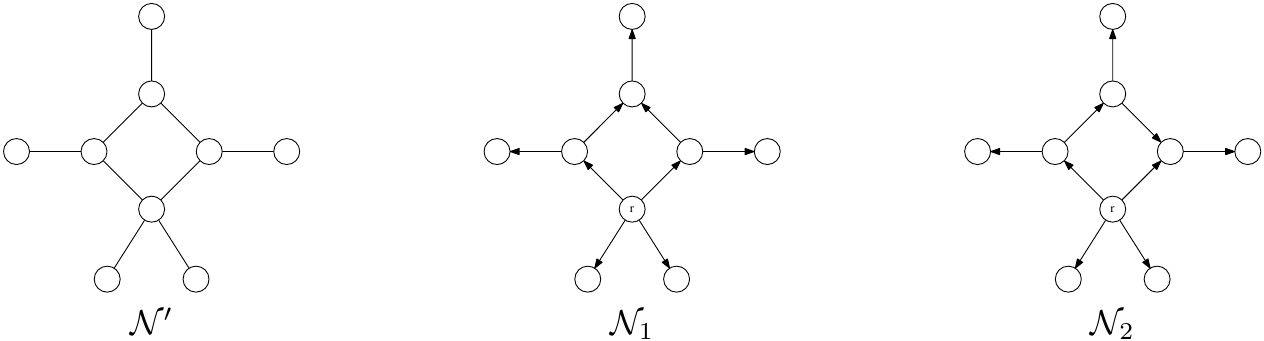}
   \caption{Networks $\N_1$ and $\N_2$ represent different orientations of the same undirected network $\N'$, i.e., $\N_1^\star=\N_2^\star=\N'$.}
   \label{fig:orient}
   \end{center}
 \end{figure}

\bibliographystyle{splncs}
\bibliography{phylo.bib} 

\clearpage
\newpage
\section*{Appendix. SAGE Code}

Below we provide a SAGE code for the function $\texttt{FGH}(N,v)$, 
which computes the triple of functions\\
$\left[F_v(x),G_v(x),H_v(x)\right]$ for a given cactus network $N$ and its vertex $v$.

\begin{footnotesize}
\begin{verbatim}
# Function FGH( N, v, true_leaves )
# Input:
#   N is a cactus network, possibly a subnetwork of the initial network
#   v is a vertex in N
#   true_leaves is an optional parameter, used internally,
#               set of vertices that are leaves in the initial network;
#               should be omitted when N is the initial network
# Output: [ F_v(x), G_v(x), H_v(x) ]

from itertools import combinations

def FGH( N, v, true_leaves = None ):

  # test if N is a DAG
  if not N.is_directed_acyclic(): 
    raise ValueError, "Input graph is not a network!"

  if true_leaves is None:
    # N is a parent network, so we initialize true_leaves: 
    true_leaves = Set(N.sinks())

  # define g.f. variable x
  x = QQ['x'].0

  if N.out_degree(v) == 0:
    # v is a leaf
    if v in true_leaves:
      return [x,0,0]
    else:
      return [0,0,1]

  P = Poset(N)

  # look for a sink in N_v if it exists
  t = v
  S = [ Set(P.principal_upper_set(u)) for u in N.neighbors_out(v) ]
  for u in combinations(S,2):
    T = u[0].intersection(u[1])
    if not T.is_empty():
      # sink is found
      t = P.subposet(T).minimal_elements()[0]
      break

  if t == v:
    # v is a regular vertex
    # compute FGH at the children of v
    cFGH = [ FGH(N,u,true_leaves) for u in N.neighbors_out(v) ]
    # apply Theorem 3
    H = prod(fgh[0]+fgh[2] for fgh in cFGH)
    G = sum((fgh[0]+fgh[1])*H/(fgh[0]+fgh[2]) for fgh in cFGH)
    F = x*prod(fgh[0]+fgh[2] + (fgh[0]+fgh[1])/x for fgh in cFGH) - x*H - G;
\end{verbatim}

\newpage

\begin{verbatim}
  else:
    # v is a source and t is a sink
    t_FGH = FGH(N, t, true_leaves)

    p = N.neighbors_in(t)
    if len(p) != 2:
      raise ValueError, "N is not a cactus network"

    N.delete_edge(p[0],t)
    # N represents L at this point
    L_FGH = FGH(N, v, true_leaves)
    N.add_edge(p[0],t)

    N.delete_edge(p[1],t)
    # N represents R at this point
    R_FGH = FGH(N, v, true_leaves)

    N.delete_edge(p[0],t)
    # N represents N_{s\t} U N_t at this point
    Nst_FGH = FGH(N, v, true_leaves)

    # remove edges in the branching path v => t from N to form B
    for q in p:
      u = q
      while u!=v:
        w = N.neighbors_in(u)
        if len(w)!=1:
           raise ValueError, "Parents error w"
        N.delete_edge(w[0],u)
        u = w[0]
    # N represents B at this point
    B_FGH = [ FGH(N, u, true_leaves) 
              for u in Set(P.closed_interval(v,p[0])) + Set(P.closed_interval(v,p[1])) ]
    B_prodH = prod(fgh[2] for fgh in B_FGH)
    # apply Theorem 9
    H = L_FGH[2] + R_FGH[2] - Nst_FGH[2] * (t_FGH[0] + t_FGH[2])
    G = L_FGH[1] + R_FGH[1] - Nst_FGH[1] * (t_FGH[0] + t_FGH[2]) 
        - (t_FGH[0] + t_FGH[1]) * B_prodH
    F = L_FGH[0] + R_FGH[0] - Nst_FGH[0] * (t_FGH[0] + t_FGH[2]) 
        - (t_FGH[0] + t_FGH[1]) * ( prod((fgh[0]+fgh[1])/x + fgh[2] for fgh in B_FGH) - B_prodH )
  return [F,G,H]
\end{verbatim}
\end{footnotesize}

\end{document}